\documentclass{amsart}

\usepackage{amsmath}
\usepackage{amssymb}
\usepackage{amsthm}
\usepackage{amsfonts}
\usepackage{amstext}
\usepackage{amsopn}
\usepackage{amsxtra}
\usepackage{mathrsfs}
\usepackage{color}      % to display colored text
\usepackage{verbatim}   % for "comment" environment
\usepackage{pdfsync}

\newtheorem{lemma}{Lemma}[section]
\newtheorem{theorem}{Theorem}[section]
\newtheorem{proposition}[lemma]{Proposition}

\newtheorem{corollary}[lemma]{Corollary}
\theoremstyle{definition}
\newtheorem{definition}[lemma]{Definition}
\theoremstyle{definition}
\newtheorem{remark}[lemma]{Remark}
\theoremstyle{definition}

\numberwithin{equation}{section}

\newcommand{\be}{\begin{equation}}
\newcommand{\ee}{\end{equation}}

\newcommand{\re}{\mathrm{Re}}
\newcommand{\im}{\mathrm{Im}}
\newcommand{\e}{\varepsilon}

\newcommand{\bbR}{{\mathbb R}}

\newcommand{\vp}{\varphi}

\newcommand{\eps}{\epsilon}

\newcommand{\bcr}{\begin{color}{red}}
\newcommand{\ec}{\end{color}\ }

\DeclareMathOperator{\diver}{div}

\def\R{{\mathbb R}}
\def\C{{\mathbb C}}
\begin{document}

\title[Dynamics of Bohmian measures]{On the dynamics of Bohmian measures}

\author[P. Markowich]{Peter Markowich}
\address[P. Markowich]{Department of Applied Mathematics and Theoretical
Physics\\
CMS, Wilberforce Road\\ Cambridge CB3 0WA\\ United Kingdom\\
and Department of Mathematics, College of Science, King Saud University
Riyadh, KSA, and Faculty of Mathematics, University
of Vienna, Nordbergstra§e 15, A-1090 Vienna, Austria}
\email{p.markowich@damtp.cam.ac.uk}
\author[T. Paul]{Thierry Paul}
\address[T. Paul]{CNRS and CMLS,\ Ecole Polytechnique\\ 
91 128 Palaiseau cedex\\ France}
\email{paul@math.polytechnique.fr}
\author[C. Sparber]{Christof Sparber}
\address[C. Sparber]
{Department of Mathematics, Statistics, and Computer Science\\
University of Illinois at Chicago\\
851 South Morgan Street
Chicago, Illinois 60607, USA}
\email{sparber@uic.edu}

\begin{abstract}
%We revisit the concept of Bohmian measures, recently introduced by the authors. 
The present work is devoted to the study of dynamical features of Bohmian measures, recently introduced by the authors.
We rigorously prove that for sufficiently smooth wave functions the corresponding Bohmian measure 
furnishes a distributional solution of a nonlinear Vlasov-type equation. Moreover, we study the associated defect measures appearing in the classical limit. In one space dimension, this yields a new 
connection between mono-kinetic Wigner and Bohmian measures. In addition, we shall study the dynamics of Bohmian measures associated to so-called semi-classical wave packets. For these type of 
wave functions, we prove local in-measure convergence of a rescaled sequence of Bohmian trajectories towards the classical Hamiltonian flow on phase space. 
Finally, we construct an example of wave functions whose limiting Bohmian measure is not mono-kinetic 
but nevertheless equals the associated Wigner measure. 
\end{abstract}

\date{}

\subjclass[2000]{81S30, 81Q20}
\keywords{Quantum dynamics, Bohmian mechanics, classical limit, kinetic equations, Wigner measure, quantum trajectories, semiclassical wave packet}
\thanks{This publication is based on work supported by Award No. KUK-I1-007-43, funded by the King Abdullah University of Science and Technology (KAUST). C.S. has been supported 
by the Royal Society via his University research fellowship and P.M. by his Royal Society Wolfson Research Merit Award. In addition, P.M. acknowledges support from the Deanship 
of Scientific Research at King Saud University in Riyadh for funding this work through the research group project NoÓRGP- VPP-124 and he also wants to thank the CNRS, since large parts of 
this work were completed while P.M. was  ``chercheur invit\'e" at the ENS Paris during the spring 2010.}

\maketitle

%\tableofcontents

\section{Introduction}
We consider the time-evolution of quantum mechanical wave functions
$\psi^\e (t, \cdot) \in L^2(\R^d; \C)$ governed by the Schr\"odinger equation:
\begin{equation}
\label{sch}
i\e \partial_t  \psi^\e  = -\frac{\e^2}{2}\Delta \psi^\e +
V(x)\psi^\e,\quad
\psi^\e(t=0,x)   = \psi^\e_{0} \in L^2(\R^d),
\end{equation}
where $x \in \R^d$, $t\in \R $,  and $V\in L^\infty(\R^d;\R)$ a given bounded potential (satisfying some additional regularity assumptions given below). In addition, we have rescaled all physical parameters such that only one 
\emph{semi-classical parameter} $0<\e \leq 1$ remains. We shall from now on assume that $\| \psi^\e_0\|_{L^2}=1$, which is henceforth propagated in time, i.e.
\be\label{masscon}
 \| \psi^\e(t) \|_{L^2} =  \| \psi^\e_0 \|_{L^2}= 1.
\ee
In addition, we also have \emph{conservation of energy}, i.e. 
\[
E^\e(t) : =   \frac{\e^2}{2} \int _{\R^d} | \nabla \psi^\e (t,x) |^2 dx + \int_{\R^d} V(x) |  \psi^\e (t,x) |^2 dx  =E^\e(0).
\]
Throughout this work, we shall assume that 
\be\label{assen}
\sup_{0<\e \leq 1} E^\e(0) < +\infty.
\ee
In other words, we assume $\psi_0^\e$ to have bounded initial energy, uniformly in $\e$.
In view of \eqref{masscon}, 
one can define out of $\psi^\e(t,x)\in \mathbb C$ real-valued probability densities from which one computes expectation values of physical observables. Possibly, the two most important such densities are the 
\emph{position} and the \emph{current-density}, given by
\begin{equation}\label{densities}
\rho^\e(t,x)= |\psi^\e(t,x)|^2, \quad J^\e(t,x) = \e \im\big(\overline{\psi^\e}(t,x)\nabla \psi^\e(t,x)\big).
\end{equation}
Already in 1926, the same year in which Schr\"odinger exhibited the eponymous equation, it has been realized by Madelung \cite{mad} that these densities can be used to rewrite \eqref{sch} in
hydrodynamical form. The corresponding \emph{quantum hydrodynamic system} reads 
\begin{equation}
\label{qhd}
\left \{
\begin{aligned}
& \, \partial_t \rho^\e + \diver J^\e = 0,\\
& \, \partial_t J^\e + \diver \left(\frac{J^\e \otimes J^\e}{\rho^\e} \right) + \rho^\e \nabla V = \frac{\e^2}{2} \rho^\e  \nabla \left( \frac{\Delta \sqrt{\rho^\e}}{ \sqrt{\rho^\e}} \right).
\end{aligned}
\right. 
\end{equation}
More precisely, it can be proved that under sufficient regularity assumptions on $V(x)$ and $\psi^\e(t,x)$, each of the nonlinear terms arising in this system is well-defined in 
the sense of distributions, see \cite[Lemma 2.1]{GaMa}. 

The quantum-hydrodynamic system \eqref{qhd} can also be seen as the starting point of \emph{Bohmian mechanics} \cite{Bo1, Bo2}. In this theory, one defines an $\e$-dependent flow-map
\[
X^\e_t: \R\times \R^d \to \R^d ; \quad  x \mapsto X^\e(t,x)
\]
via the following differential equation
\begin{equation*}\label{bohm1}
\dot X^\e(t, x) = u^\e(t,X^\e(t, x)) ,\quad X^\e(0, x) = x\in \R^d,
\end{equation*} 
where the vector field $u^\e$ is (formally) given by
\begin{equation*}\label{velocity}
u^\e(t,x):= \frac{J^\e(t,x)}{\rho^\e(t,x)}= \e \im \left(\frac{ \nabla \psi^\e(t, x)}{ \psi^\e(t, x)} \right ) 
\end{equation*}
and the initial data is assumed to be distributed according to the measure $\rho_0^\e (x) \equiv |\psi^\e_0(x)|^2$. Note that in terms of 
$\rho^\e$ and $J^\e$, the kinetic energy reads
\[
E_{kin}: = \frac{\e^2}{2} \int _{\R^d} | \nabla \psi^\e (x) |^2 dx =\frac{1}{2}  \int _{\R^d} \frac{|J^\e(x)|^2}{\rho^\e(x)} \, dx +  \frac{\e^2}{2} \int _{\R^d}  | \nabla \sqrt{\rho^\e} |^2 dx. 
\]
From energy conservation, we therefore conclude $u^\e \in L^2 (\R^d, \rho^\e dx)$, but not necessarily continuous. Moreover, it
has been rigorously proved in \cite{BDGPZ} (see also \cite{TeTu}) that  $X^\e(t,\cdot)$ is for all $t\in \R$ well-defined $\rho_0^\e$ - $a.e.$ and that 
$$\rho^\e(t,x) = X^\e_t  \, \# \,  \rho^\e_0(x),$$ i.e. $\rho^\e(t,x)$ is the \emph{push-forward} of the initial density $ \rho_{0} ^\e(x)$ under the mapping $X^\e_t: x \mapsto X^\e(t,x)$, see the definition in Equation \eqref{pf1} below.

Whereas, the above dynamics $ x \mapsto X^\e(t,x)\in \R^d$ is stated in physical space, 
Bohmian mechanics can also be reformulated on {\it phase space} $\R^d_x\times \R^d_p$, by using the concept of \emph{Bohmian measures}, recently introduced by the authors in \cite{MPS}: 
\begin{definition} \label{defbohm} 
Let  $\e>0$. For  $\psi^\e \in H^1(\R^d)$, with associated densities $\rho^\e, J^\e$, given by \eqref{densities}, the  
\emph{Bohmian measure} $\beta^\e\equiv \beta^\e[\psi^\e] \in \mathcal M^+(\R^d_x \times \R^d_p)$ is defined by
$$
\langle  \beta^\e , \varphi \rangle := \int_{\R^d} \rho^\e(x)  \varphi \left(x, \frac{J^\e(x)}{\rho^\e(x)} \right) dx , \quad \forall \, \varphi \in C_0(\R^d_x \times \R^d_p),
$$
where $C_0(\R^d_x \times \R^d_p)$ denotes the space of continuous function vanishing at infinity and $\mathcal M^+(\R^d_x \times \R^d_p)$ the set of non-negative Radon measures on phase space.
\end{definition}
It has been shown shown in \cite{MPS} that if $\psi^\e(t,x)$ solves \eqref{sch}, then the corresponding Bohmian measure $\beta^\e(t,x,p)$  is the push-forward of 
\be\label{ini}
\beta^\e[\psi^\e_0 ]\equiv \beta^\e_0(x,p)=\rho^\e_0 (x) \delta(p -u_0^\e(x)),
\ee
under the $\e$-dependent phase space flow 
\be \label{phaseflow} 
 \Phi^\e_t: (x,p) \mapsto (X^\e(t,x,p),P^\e(t,x,p))
\ee induced by 
\begin{equation}
\label{bohm2}
\left \{
\begin{aligned}
& \,  \dot X^\e= P^\e , \ \\
& \, \dot P^\e= - \nabla V(X^\e) -   \nabla V^\e_B(t, X^\e) , 
\end{aligned}
\right. 
\end{equation}
where $V^\e_B(t,x)$, denotes the so-called \emph{Bohm potential} (see \cite{DuTe, fs, tp}):
\begin{equation}
\label{bohmpot}
V^\e_B(t,x):= -\frac{\e^2}{2} \frac{\Delta \sqrt{\rho^\e(t,x)}}{ \sqrt{\rho^\e(t,x)}}.
\end{equation}
More precisely, under mild regularity assumptions on $V$, the flow $\Phi^\e_t$ is shown to exists globally in time 
for almost all $(x,p) \in \R^{2d}$, \emph{relative to the measure} $\beta_0^\e$ and is continuous in time on its maximal open domain, cf. \cite[Lemma 2.5] {MPS}. 
Note that the specific form of the initial data \eqref{ini} implies that the phase-space flow $\Phi^\e_t$ is projected 
onto the graph of $u^\e_0$, i.e.
$$
\mathcal L^\e: = \{ (x,p) \in \R^d_x\times \R^d_p : p = u^\e_0( x) \},
$$
whose time-evolution is governed by the Bohmian flow \eqref{bohm2}. \\

The fact that $\beta^\e(t) =   \Phi^\e_t  \, \# \,  \beta^\e_0(x)$, is usually called \emph{equivariance} of Bohmian measures (see \cite{DuTe}) and makes $\beta^\e(t)$ 
a natural starting point for the investigation of the classical limit as $\e \to 0_+$ of Bohmian mechanics. 
In \cite{MPS} we gave an extensive study (invoking Young measure theory) of the possible oscillation and concentration phenomena appearing in $\beta^\e$ as $\e \to 0_+$ and compared our findings to the 
by now classical theory of \emph{Wigner measures}, cf. \cite{GMMP, LiPa, SMM}. One thereby associated to any wave function $\psi^\e$ a \emph{Wigner function} $w^\e[\psi^\e]\equiv w^\e$, defined by \cite{Wi}:
\begin{equation*}
w^\e (t,x,p): = \frac{1}{(2\pi)^d} \int_{\R^d}
\psi^\e\left(t,x-\frac{\e}{2}y 
\right)\overline{\psi^\e} \left(t,x+\frac{\e}{2}y
\right)e^{i y \cdot p}\,  dy. 
\end{equation*}
It is well known that although $w^\e(t,x,p)\not \geq0$ in general, it admits as $\e \to 0_+$ a weak limit $w(t) \in \mathcal M^+(\R^d_x \times \R^d_p)$, usually called \emph{Wigner measure} (or semi-classical measure). 
The latter is known to give the possibility to  describe in a ``classical'' manner the expectation values of physical observables via
\begin{equation*}
\lim_{\e\to 0} \langle \psi^\e(t), \text{Op}^\e(a)\psi^\e(t)\rangle_{L^2} =  \iint_{\R^{2d}} a(x,p) w(t, x,p) dx \, dp,
\end{equation*}
where the physical observable $ \text{Op}^\e(a)$ is a self-adjoint operator obtained from the classical symbol $a\in C^\infty_{\rm b}(\R^d_x \times \R^d_p)$ through Weyl-quantization, see, e.g., \cite{GMMP, SMM} for a precise definition.

Similarly to that, we were able to establish in \cite{MPS} the existence of a limiting non-negative phase space measure $ \beta (t) \in \mathcal M^+(\R^d_x \times \R^d_p)$, such that, after extracting an 
appropriate sub-sequence (denoted by the same symbol):
$$
\beta^\e  \stackrel{\e\rightarrow 0_+
}{\longrightarrow} \beta \quad \text{in $L^\infty(\R_t; \mathcal M^+(\R^d_x \times \R^d_p)) \, 
%{\rm w}-\ast$}.
{\rm weak}^\ast$}.
$$
If, in addition, $\psi^\e(t)$ is \emph{$\e$-oscillatory}, i.e. 
\be \label{cond}
\sup_{0<\e \le 1}( \| \psi^\e(t)  \|_{L^2} +   \| \e \nabla \psi^\e(t)  \|_{L^2}) < + \infty.
\ee
one can prove (cf. \cite[Lemma 3.2]{MPS}) that the limiting phase space measure $\beta(t)$ incorporates the classical limit of the {position}
and current density in the sense that
\begin{equation}\label{limdensities}
\rho^\e(t,x)  \stackrel{\e\rightarrow 0_+
}{\longrightarrow} \int_{\R^d} \beta(t, x,dp) , 
\end{equation}
and 
\begin{equation}\label{limdensities1}
J^\e(t,x)  \stackrel{\e\rightarrow 0_+
}{\longrightarrow} \int_{\R^d} p \beta(t, x,dp) .
\end{equation}
Hereby the limits have to be understood in $ \mathcal M^+(\R^d_x) \, 
%{\rm w}-\ast$,
{\rm weak}^\ast$, uniformly on compact time-intervals $I\subset \R_t$. Note that condition \eqref{cond} is satisfied \emph{for all} $t\in \R$, in view of \eqref{masscon},
the conservation of energy and our initial assumption \eqref{assen} (since for any $V\in L^\infty(\R^d)$ we can be, assume without loss of generality, $V(x)\ge 0$). This is the reason to 
impose \eqref{masscon} and \eqref{assen}, throughout this work. In summary, the limiting Bohmian measure $\beta (t)$ 
yields the classical limit of the quantum mechanical position and current densities, by taking the zeroth and first moment with respect to $p\in \R^d$, analogous to the case of Wigner measures. 

While \cite{MPS} establishes several links between Wigner and Bohmian measures, it is {\it not} concerned with the dynamical properties of $\beta^\e(t)$ and its corresponding limit $\beta(t)$ 
(except for a short discussion on connections to WKB analysis before caustics, see  \cite[Proposition 6.1]{MPS} and Remark \ref{WKBrem} below). 
In contrast to that, the goal of the present work is do establish several {\it dynamical} result. The theorems obtained below can therefore be seen as independent from 
\cite{MPS}, except for the basic existence and weak convergence results, recalled above.

\section{Main results}\label{mres}
The main objective of the current work is to analyze the dynamics of $\beta^\e(t)$. Note that the equivariance property strongly suggests that $\beta^\e(t)$ satisfies the following 
nonlinear equation (of Vlasov-type):
\begin{equation}\label{maineq}
\left \{
\begin{aligned}
& \partial_t \beta^\e + p \cdot \nabla_x \beta^\e - \nabla_x \left(V  -\frac{\e^2}{2} \frac{\Delta_x \sqrt{\rho^\e}}{ \sqrt{\rho^\e}}\right) \cdot \nabla_p \beta^\e = 0,\\
& \rho^\e(t,x) = \int_{\R^d} \beta^\e(t,x,dp),
\end{aligned}
\right.
\end{equation}
subject to initial data $\beta_0^\e(x,p)$ given by \eqref{ini}. Note that the system \eqref{maineq} can be written in one
line as
\begin{equation*}\label{oneline}
\partial_t \beta^\e + p \cdot \nabla_x \beta^\e - \nabla_x \left(V 
- \frac{\e^2}{2}\frac{\Delta_x  \sqrt{\int_{\R^d} \beta^\e dp}}{\sqrt{\int_{\R^d} \beta^\e dp}} \right) \cdot \nabla_p \beta^\e = 0,
\end{equation*}
Due to the strong nonlinear nature of this equation and in particular due to possible singularities at points $x\in \R^d$ where $\rho^\e(t,x) = 0$, it is 
a non-trivial task to rigorously prove that $\beta^\e(t)$ furnishes a distributional solution to \eqref{maineq}. It will be one of the goals of this work to show that this is indeed the case. To this end, we shall derive 
appropriate bounds for all terms arising in the weak formulation of \eqref{maineq} after being evaluated at Bohmian measures. This rigorously establishes the existence of a nonlinear evolution equation for $\beta^\e(t)$ in a similar spirit 
as the results of \cite{GaMa} for the quantum hydrodynamic system \eqref{qhd}. More precisely, the first main result of this work is as follows:

\begin{theorem}\label{main}
Let $V\in C^1_{\rm b}(\R^d;\R)$ and $\psi^\e_0\in H^3(\R^d)$ with corresponding $\rho_0^\e, J_0^\e$ defined by \eqref{densities}. Then, for all $\e >0$, the Bohmian measure
$$\beta^\epsilon(t,x,p) = \rho^\e(t,x) \delta \left(p - \frac{J^\e(t,x)}{\rho^\e(t,x)}\right) ,$$
is a weak solution of \eqref{maineq} in $\mathcal D'(\R_t \times \bbR^d_x\times \bbR^d_p)$ and in $\mathcal D'([0,\infty) \times \bbR^d_x\times \bbR^d_p)$ with initial data \eqref{ini}.
\end{theorem}

In a second step we shall study the classical limit of equation \eqref{maineq}. To this end, let us recall the following definition used in \cite{MPS}:
\begin{definition}[Mono-kinetic measure] \label{mk}
A measure 
 $\mu \in \mathcal M^+(\R^d_x\times\R^d_p)$ is said to be \emph{mono-kinetic}, if there exists a $\rho \in \mathcal M^+(\R_x^d)$ and a function $u(x)$ defined $\rho$ - $a.e.$, such that 
$$
\mu(x,p) = \rho(x) \, \delta (p-u(x)).
$$
\end{definition}
Mono-kinetic phase space measures define velocity distributions, which are given by the graph of $u(x)$ 
and are thus particularly interesting for our purposes since, clearly, $\beta^\e(t)$ is mono-kinetic by definition, its limit  
however, will not be in general. 

To get further insight, we pass to the limit $\e \to 0_+$ (after extraction of a subsequence) 
in the first three linear terms of equation \eqref{maineq}. This naturally leads to the following definition of a possible \emph{defect} $\mathcal F(t,x,p)\in \R^d$: Along a chosen sub-sequence $\{ \e_n \}_{n \in \mathbb N}$, let 
\begin{equation}\label{F}
\mathcal F:=  \lim_{\e \to 0_+}  (\diver_p (\nabla_x V^\e_B \beta^\e)), \quad \mbox{in $\mathcal D'(\R_t\times \R^d_x\times \R^d_p)$,}
\end{equation}
such that
\begin{equation*}
 \partial_t \beta + p \cdot \nabla_x \beta - \nabla_x  V  \cdot \nabla_p \beta = \mathcal F ,
\end{equation*}
A partial characterization of the defect $\mathcal F$, will be given in Section \ref{sec: def}. In particular, it leads to the following result:

\begin{theorem}\label{thneu} Let $d=1$ and $\beta^\e(t)$ solve \eqref{maineq}. In addition assume that at $t=0$, the limiting measure satisfies
\[
\beta_0(x,p) = w_0(x,p) = \rho_0(x) \delta(p -u_0(x)),
\]
where $w_0(x,p)$ is the Wigner measure associated to $\psi^\e_0(x)$. Then, on any time-interval $I\subseteq \R_t $ on which $w(t)$ is mono-kinetic, it holds
\[ 
\beta(t,x,p) = w(t,x,p)=\rho(t,x) \delta(p -u(t,x)),
\]
in the sense of measures.
\end{theorem}

In contrast to \cite{MPS}, where we established certain sufficient criteria for $\psi^\e(t)$, yielding mono-kinetic Wigner and/or limiting Bohmian measures, this result directly gives $\beta(t) = w(t)$ from the fact that 
the Wigner measure is mono-kinetic. It is well known, cf. \cite{GaMa, SMM} that the Wigner measure is generically mono-kinetic
before \emph{caustics}, i.e.~before the appearance of the first shock in the (field driven) Burgers equation
\be\label{burg}
\partial_t u + (u\cdot \nabla) u + \nabla V (x) = 0.
\ee
Note however, that there are situations in which $w$ is of the form given in Theorem \ref{thneu}  even \emph{after} caustics, see e.g. Example 4 given in \cite{GaMa} and 
Example 1 in \cite{SMM}. Both ot these examples furnish so-called \emph{point caustics}, i.e. caustics where all the rays of geometric optics cross at one point. The result given above therefore 
shows that in some situations the limiting Bohmian measure $\beta(t)$ can indeed give the correct classical limit (for all physical observables) \emph{even after caustics}. 
In addition, Theorem \ref{thneu} generalizes \cite[Proposition 6.1]{MPS}, at least in $d=1$ spatial dimensions.

\begin{remark}\label{WKBrem} Equation \eqref{burg} can be seen as the formal limit obtained from \eqref{sch} by means of WKB analysis. One thereby seeks solution to \eqref{sch} in the following form 
$$\psi^\e(t,x) = a^\e(t,x)e^{i S(t,x)/\eps}, 
$$
with (real-valued) $\e$-independent phase $S$ and (possibly complex-valued) amplitude $a^\e\sim a_0 + \e a_1 +\e^2 \dots$, in the sense of asymptotic expansions.
Plugging this ansatz into \eqref{sch} and neglecting terms of order $\mathcal O(\e^2)$ yields \eqref{burg} upon identifying $u  = \nabla S$, cf. \cite{SMM} for more details (see also Section 6 of \cite{MPS} where the 
connection between WKB analysis and Bohmian measures is briefly discussed). 
\end{remark}

Finally, we shall consider the particular case where the initial data $\psi^\e_0$ is a so-called \emph{semi-classical wave packet}, see Theorem \ref{th2} below. The classical limit of Bohmian trajectories in this particular situation 
has been recently studied in \cite{DuRo}. There it has been shown that the Bohmian trajectories $X^\e(t,x)$ converge (in some suitable topology) to $X(t)$, the classical particle trajectory induced by the Hamiltonian system 
\begin{equation}
\label{classical}
\left \{
\begin{aligned}
& \,  \dot X= P , \quad X(0) = x_0, \\
& \, \dot P= - \nabla V(X),\quad  P(0) = p_0.
\end{aligned}
\right. 
\end{equation}
One should note that the mathematical methods used in \cite{DuRo} are rather different from ours and that no convergence result for $P^\e(t)$ is given, 
except for $p_0=0$ (and, as a variant, for a class of time-averaged Bohmian momenta). 
In comparison to that, the use of $\beta^\e(t)$, together with classical Young measure theory, 
allows us to conclude the following result:
\begin{theorem} \label{th2}
Let $V\in C_{\rm b}^3(\R^d)$. In addition, assume that 
$\psi_0^\e$ is of the following form
$$
\psi^\e_0(x) = \e^{-d/4}  a\left( \frac{x-x_0}{\sqrt \e} \right) e^{i p_0 \cdot x /\e}, \quad x_0, p_0 \in \R,
$$
with given $\e$-independent amplitude $a\in \mathcal S(\R^d;\mathbb C)$, such that $|a(x)|^2 > 0$ a.e. on $\R^d$. 
\begin{enumerate} 
\item Then, the limiting Bohmian measure satisfies
$$
\beta(t,x,p) = w(t,x,p) =\| a\|_{L^2}^2 \, \delta(x-X(t)) \delta (p -P(t)),
$$
where $X(t), P(t)$ are the classical trajectories defined by \eqref{classical}.

\item Consider the following re-scaled Bohmian trajectories 
$$
Y^\e (t,y) = X^\e (t, x_0 + \sqrt{\e} y), \quad Z^\e (t,y) = P^\e (t, x_0 + \sqrt{\e} y),
$$
where $X^\e(t, x), P^\e(t,x)$ solve \eqref{bohm2} with initial data $(x_0 ,p_0)\in\R^{2d}$. Then
$$
Y^\e  \stackrel{\e\rightarrow 0_+ }{\longrightarrow} X, \quad Z^\e  \stackrel{\e\rightarrow 0_+ }{\longrightarrow} P,
$$
locally in measure on $\R_t \times \R^d_x$, i.e., for every $\delta>0$ and every 
Borel set $\Omega \subseteq \R_t \times \R^d_x$ with finite Lebesgue measure $\mathbb L$, it holds
\be\label{defcon2}
\lim_{\e\to 0} \mathbb L\big(\{(t,y)\in \Omega:\ | (Y^\e (t,y),Z^\e (t,y))-(X(t),P(t))| \ge \delta \} \big) = 0.
\ee
%where $\mathbb L$ denotes the Lebesgue measure.
\end{enumerate}
\end{theorem} 
Let us compare now our results with the corresponding ones in \cite{DuRo}: Besides the fact that we are able to infer 
convergence of the Bohmian momentum $P^\e(t)$, the topology we use is different from the one used there. Rephrased in our notation \cite{DuRo} proves that, for all $T > 0$ and $\gamma > 0$ there exists some $R < \infty$
such that, for all $\e$ small enough,
\be\label{them}
\mathbb P_{\rho_0^\e}(\{x_0\in \R^3 |\max_{t\in[0,T]}\vert X^\e(t,x_0)-X(t)\vert\leq R\sqrt\e\})>1-\gamma,
\ee
where $\mathbb P_{\rho^\e_0}$ is the probability measure with density $\rho^\e_0=\vert \psi^\e_0(x)\vert^2$. 
We see clearly that the comparison between the two results is not straightforward, as \eqref{them} 
measures the ``good" points 
and \eqref{defcon2} the ``bad" ones. 
Moreover, \eqref{them} is more precise for finite times, whereas \eqref{defcon2} doesn't
require a-priori bounds on the time dependence and additionally implies almost everywhere convergence of sub-sequences, see Remark \ref{rem: last}. 
It would certainly be interesting to study the link between the two approaches more precisely.
\begin{remark}
In view of Theorem \ref{thneu} and Theorem \ref{th2} one might guess that $w=\beta$ only if both are mono-kinetic phase space distributions. 
This is consistent with the examples given in \cite{MPS} but still wrong in general. 
Indeed, in the appendix of the present work, we shall construct a family of wave functions $\psi^\e$ for which $w = \beta$ is absolutely 
continuous with respect to the Lebesgue measure on $\R^d_p$. %This henceforth closes a gap in the case studies given in \cite{MPS}.
\end{remark}

The paper is now organized as follows: In the upcoming section we collect some basic mathematical estimates needed throughout this work. 
In Section \ref{sec: weak} we establish the fact that $\beta^\e(t)$ indeed furnishes a distributional solution of \eqref{maineq}. 
In Section \ref{sec: def} the defect $\mathcal F$ will be partially characterized, yielding the proof of Theorem \ref{thneu}. In Section \ref{sec: coh}, the particular case of semi-classical wave packets will be studied including the 
proof of Theorem \ref{th2}. 
Finally, in Appendix A we shall present and example in which $w=\beta$ but not mono-kinetic.

\section{Static estimates}\label{sec: static}

In order to establish the weak formulation of \eqref{maineq}, we shall heavily rely on the following static, i.e. time-independent, estimate.

\begin{proposition}\label{te}
Let $\psi\in C^2(\R^d)$, then it holds:
$$
\left\vert \diver\left(\nabla\vert\psi\vert\otimes\nabla\vert\psi\vert\right)\right\vert\leq 
d \left(\sup_{\ell jk}\vert\partial_\ell\partial_j\psi\partial_k\psi\vert+
\frac 12\left\vert\im\left(\frac{\nabla\psi}{{\psi}}\right)\right\vert
\sup_{\ell j}\vert\partial_\ell\psi\partial_j\psi\vert\right),
$$
with $j, \ell, k =1, \dots, d$.
\end{proposition}

\begin{proof}

We denote $\partial_j:=\partial_{x_j}$ and first compute
\[
\partial_j|\psi|=\frac{\overline\psi \partial_j\psi+\psi\overline{\partial_j\psi}}{2(\psi\overline\psi)^{1/2}},
\]
which yields 
\[
\partial_\ell\vert\psi\vert \, \partial_j\vert\psi\vert =\re \left(\partial_\ell\psi\overline{\partial_j\psi}+\frac{\overline\psi}\psi
\partial_\ell\psi\partial_j\psi\right),
\]
Using this we can write
\be\label{form}
\diver\left(\nabla\vert\psi\vert\otimes\nabla\vert\psi\vert\right)=
  \sum_{k=1}^d \re \ \partial_k \left(\partial_\ell\psi\overline{\partial_j\psi}+\frac{\overline\psi}\psi
\partial_\ell\psi\partial_j\psi\right),
 \ee
 where each term in this series will be estimated separately (and in the same way). 
 In order to handle terms in which the partial derivative $\partial_k$ acts on $\overline\psi/\psi$ we note that
\[
\partial_k\left(\frac{\overline\psi}\psi\right)=\frac{\partial_k\overline\psi}\psi-\frac{\partial_k\psi\overline\psi}{\psi^2}
=\left(\frac{\partial_k\overline\psi}{\overline\psi}-\frac{\partial_k\psi}\psi\right)\frac{\overline\psi}\psi,
\]
and we henceforth obtain
\[
\partial_\ell\psi\partial_j\psi\partial_k\left(\frac{\overline\psi}\psi\right)=\frac1
{2}\im\left(\frac{\partial_k\psi}{{\psi}}\right)\partial_\ell\psi\partial_j\psi\frac{\overline\psi}\psi.
\]
Using this on the r.h.s. of \eqref{form} and summing up all terms yields the assertion of the lemma.

\end{proof}

In the upcoming analysis, we shall use the established estimate in the following form, where we denote by $\| f \|_{\dot H^1}:=  \| \nabla f \|_{L^2} $ the usual $\dot H^1(\R^d)$ semi-norm: 

\begin{corollary}\label{tec}
Fix $\e>0$. 
Then for any $\psi^\e \in H^2_{\rm loc}(\R^d)$ and for any test function $\varphi\in\mathcal D(\bbR^d_x\times \bbR^d_p)$, we have
\be\label{semi2}
%\left\vert\int_{\R^d}  \diver\left(\nabla\sqrt{\rho^\e}\otimes\nabla\sqrt{\rho^\e}\right)\varphi\left(x,\frac {J^\e}{ %\rho^\e}\right)dx\right\vert\leq M^\e < +\infty,
\int_{\R^d} \left\vert \diver\left(\nabla\sqrt{\rho^\e}\otimes\nabla\sqrt{\rho^\e}\right)\varphi\left(x,\frac {J^\e}{ \rho^\e}\right)\right\vert dx\leq M^\e < +\infty,
\ee
where $\rho^\e, J^\e$ are defined in \eqref{densities}. Explicitly, we find
\begin{align*} 
M^\e\leq & \ \frac{d}{\e} 
\Vert\psi^\e \Vert^2_{\dot H^1(\Omega)}\sup_{\xi\in\R^d}\int_{\R^d} \vert \xi\varphi(x,\xi)\vert dx \\
& \ +\e d \Vert\psi^\e \Vert_{\dot H^1(\Omega)}\Vert \nabla \psi^\e \Vert_{\dot H^1(\Omega)}\sup_{\xi\in\R^d}\int_{\R^d}\vert\varphi(x,\xi)\vert dx,
\end{align*}
where $\Omega\subset \R_x^d$ denotes any compact set containing the support of $\varphi(.,p)$  for all $p\in\R^d$.

\end{corollary}
\begin{proof} The proof follows directly from the estimate given in Proposition \ref{te} and a density argument in $H^2_{\rm loc}(\R^d)$. \end{proof}

In order to understand how $M^\e$ behaves with respect to $\e$ we first note that due to the semi-classical scaling of the equation \eqref{sch} 
$\| \psi^\e(t)\| _{\dot H^1}$ is not uniformly bounded as $\e\to 0_+$. In fact $\psi^\e(t)$ is $\e$-oscillatory  for all $t\in \R$ and thus each derivative scales like $1/\e$. 
Invoking the conservation laws of mass and energy, we have to use the re-scaled semi-norms $\| f \|_{\dot H_\e^1}:=  \| \e \nabla f \|_{L^2} $ in order to write $M^\e$ in terms of uniformly bounded (w.r.t. $\e$) expressions.
Doing so,  we find that $M^\e = \mathcal O(1/\e^3)$ as $\e \to 0_+$.
\begin{remark}
Note, however, that if $\psi^\e$ is of WKB type, i.e. $\psi^\e(x)=a(x) e^{iS(x)/\varepsilon}$ with real-valued phase $S(x)\in \R$ 
and $\e$-independent amplitude $a$,  then the left hand side in \eqref{semi2} 
obviously is bounded as $\e\to 0$, whereas the right hand side blows up. In addition, it is easy to check that if one takes  a superposition of
two WKB states (such that $\nabla S_1 \not = \nabla S_2$), the bound \eqref{semi2} 
%is easily saturated.
cannot be improved in general.
\end{remark}

\section{Bohmian measures as distributional solutions of \eqref{maineq}} \label{sec: weak}

\subsection{Mathematical preliminaries}

Let us first note that the assumption $V\in C_{\rm b}^1(\R^d;\R)$ is sufficient to ensure that, for each $\e>0$, the Hamiltonian operator
\begin{equation*}\label{ham}
H^\e= -\frac{\e^2}{2}\Delta + V(x),
\end{equation*}
is 
%essentially 
self-adjoint on $D(H^\e)=H^2(\R^d)\subset L^2(\R^d)$. It therefore generates a unitary (strongly continuous) group $U^\e(t) = e^{-it H^\e /\e}$ on $L^2(\R^d)$, which ensures 
the global existence of a unique solution $\psi^\e(t) = U^\e(t) \psi_0$ of the Schr\"odinger equation \eqref{sch}, such that 
$$ \| \psi^\e(t, \cdot) \|_{L^2} = \| \psi^\e_0 \|_{L^2}.$$
Moreover, since $\psi^\e_0\in H^3 \subset D(H^\e)$, standard semi-group theory \cite{Pa} implies $\psi^\e(t) \in  D(H^\e)= H^2(\R^d)$, for all $t\in \R$. Clearly, this also yields that $\rho^\e(t)\in L^1(\R^d;\R)$ as well as 
$J^\e(t) \in L^1(\R^d;\R^d)$, for all $t\in \R$ and 
$$
E^\e(t) = E^\e(0)< +\infty.
$$
In addition, since $H^\e$ and $U^\e(t)$ commute, and we obtain that $$\| (H^\e)^{n/2} \psi^\e (t) \|_{L^2(\R^d)} = \| (H^\e)^{n/2} \psi^\e_0 \|_{L^2(\R^d)}.$$ Since $V\in L^\infty(\R^d)$ the latter is equivalent to the 
$n$-th Sobolev norm $$\|f \| _{H^n}:= \| (1 + |\xi |^{n/2}) \widehat f \|_{L^2},$$ and we immediately infer the following result:
\begin{lemma}\label{h3}
Under the assumptions of Theorem \ref{main}, $\psi^\e(t)\in H^3(\R^d)$ for all $t\in\R$.
\end{lemma}

With this result in hand, we are sure to be able to apply the estimates established in Section \ref{sec: static}.

\subsection{Weak formulation of \eqref{maineq}}

In order to make sense of $\beta^\e(t)$ as a weak solution of \eqref{maineq}, the main problem is to understand the weak formulation of $\nabla_x V_B^\e\cdot  \nabla_p \beta^\e=\diver_p(V^\e_B \nabla_x \beta^\e)$. 
To this end, consider a class of compactly supported test functions 
$\vp(t, x,p)=\chi(t,x)\sigma(p)$ with $\chi \in C_{\rm c}^\infty(\bbR_t \times\bbR^d_x)$, $\sigma\in C_{\rm c}^\infty(\bbR^d_p)$ and compute
\begin{align*}
&\int_0^\infty\iint_{\bbR^{2d}}\nabla_x V^\e_B\cdot \nabla_p  \vp(t,x,p) d\beta^\e(t,x,p) dt = \\
& = \int_0^\infty\int_{\bbR^d}\chi(t, x)
\nabla_x V^\e_B(t,x) \cdot \nabla\sigma(u^\e(t,x)) \rho^\e(t, x) \, dx \, dt.
\end{align*}
since, by definition, $\beta^\e(t,x,p) = \rho^\e(t,x) \delta(p-u^\e(t,x))$ where denote $u^\e:= \frac{J^\e}{\rho^\e}$, i.e. the quantum mechanical velocity field. To this end, we recall that $u^\e \in L^2 (\R^d, \rho^\e dx)$.
The following lemma then shows that 
this weak formulation indeed makes sense.

\begin{lemma}\label{ppf1}
Let $\e>0$. For $\sigma\in C_{\rm c}^\infty(\bbR^d_p)$ and $\chi \in C_{\rm c}^\infty(\bbR_t \times\bbR^d_x)$ we have
\[
\int_0^\infty\int_{\bbR^d} | \chi(x,t) | \, |
\nabla\sigma(u^\e(x,t))| \, |\nabla_x V^\e_B(x,t)| \rho^\e(t, x)\,dx\, dt <+ \infty.
\]
\end{lemma}
This result is the key for proving that Bohmian measures furnishes a distributional solution of \eqref{maineq}. 
\begin{proof}
A simple computation shows that
\begin{align*}
\rho^\e\nabla_x
V_B^\e= & \ \frac{\e^2}2\nabla \Delta \rho^\e-\frac{\e^2}4 \diver\left(\frac{\nabla\rho^\e\otimes\nabla\rho^\e}{\rho^\e}\right)\\
= & \ \frac{\e^2}2\nabla \Delta \rho^\e - \e^2 \diver\left(\nabla\sqrt{\rho^\e}\otimes\nabla\sqrt{\rho^\e}\right).
\end{align*} 
Inserting this into the weak formulation of $\nabla_x V_B^\e\cdot  \nabla_p \beta^\e$, we can estimate
\begin{align*}
&\int_0^\infty\int_{\bbR^d} | \chi(x,t) | \, |
\nabla\sigma(u^\e(x,t))| \, |\nabla_x V^\e_B(x,t)| \rho^\e(x)\, dx\, dt \\
&\leq \ \frac{\e^2}2\int_0^\infty\int_{\bbR^d} | \chi(x,t)| |\nabla\sigma(u^\e(x,t))||\nabla \Delta \rho^\e| \, dx\, dt\, + \\
& \ + \e^2\int_0^\infty\int_{\bbR^d} | \chi(x,t)| \, |\nabla\sigma(u^\e(x,t))| \, \vert\diver\left(\nabla\sqrt{\rho^\e}\otimes\nabla\sqrt{\rho^\e}\right)\vert \, dx\, dt
\end{align*}
The two terms on the right hand side can the be treated as follows: 
By Lemma \ref{h3} we have $\psi^\e \in H^3(\R^d)$ and thus $\nabla \Delta \rho^\e$ is in $L^1(\R)$, for all $ t\in\R$. Therefore
\[
\int_0^\infty\int_{\bbR^d}|\chi(t,x)| \, | \nabla\sigma(u^\e(x,t))| \, |\nabla \Delta \rho^\e| \, dx\, dt <+\infty.
\]
On the other hand, inequality \eqref{semi2} directly yields (for any fixed $\e >0$)
\[
\int_0^\infty\int_{\bbR^d}| \chi(t,x)| \, |\nabla\sigma(u^\e(x,t))| \, \vert\diver\left(\nabla\sqrt{\rho^\e}\otimes\nabla\sqrt{\rho^\e}\right)\vert \, dx \, dt<+\infty,
\]
and the assertion is proved. 
\end{proof}

As a final 
preliminary step, let us recall the classical notion of the \emph{push-forward for measures}: Let $\mu_0\in\mathcal M(\R^d)$ and $f:\R^d\to\R^d$ a measurable map. Then 
$ \mu_1=f \, \# \,  \mu_0$ is called the \emph{push-forward of $\mu_0$ under $f$}, if for every $\sigma\in C_0(\bbR^d)$, it holds: 
\be\label{pf1}
%\int_{\bbR^d}\sigma(x)\mu_1(x)dx=\int_{\bbR^d}\sigma(f(x))\mu_0(x)dx,
\int_{\bbR^d}\sigma(x)d\mu_1=\int_{\bbR^d}\sigma(f(x))d\mu_0,
\ee
By using an approximation argument this condition can be relaxed in order to take into account test-functions $\sigma$ which are (only) 
integrable with respect to $\mu_1$, but do not necessarily lie in $C_0$, see Definition 1 and below in \cite{MC}. 
\begin{remark} Indeed, the equivalent definition of a push-forward: $$\mu_1 (\Omega) = \mu_0 (f^{-1}(\Omega)),$$ for all measurable sets $\Omega\subseteq \R^d$ 
and any measurable function $f$, allows one to choose indicator functions of measurable sets as test functions.
Based on this observation, an approximation argument by simple functions leads to the fact that both sides of \eqref{pf1} 
are finite and equal for  all test functions $\sigma$ which are $\mu_1$ integrable \cite{MC}.\end{remark}
The reason for using this slightly more general version of the usual push-forward formula will become clear in the proof given below. 
We first note that
\begin{align*}
&\int_0^\infty\iint_{\bbR^{2d}} \vp(t,x,p) d\beta^\e(t, dx,dp) dt  =\int_0^\infty\int_{\bbR^{d}}\vp(t,X^\e(t,x),u^\e(t,x)) \rho_0^\e(x) dx \, dt,
\end{align*}
by using the fact that $\beta^\e(t,x,p)$ is the push-forward of the measure $\rho_0^\e(x) \delta(p-u_0(x))$ under the Bohmian phase space flow 
$\Phi^\e_t$ defined in \eqref{phaseflow}.

\begin{proof}[Proof of Theorem \ref{main}] Let $\vp \in C_{\rm c}^\infty([0,+\infty)\times \R^d_x\times\R_p^d)$ be a compactly supported test-function such that $\varphi(t,x,p)=\chi(t,x)\sigma(p)$. Then, the 
weak formulation of \eqref{maineq} reads
\begin{align*}
&  \int_0^\infty \iint_{\bbR^{2d}} \sigma(p)( (\partial_t \chi(t,x)  + p\cdot \nabla_x \chi(t,x))  \beta^\e(t, dx, dp) dt \\
&  - \int_0^\infty \iint_{\bbR^{2d}} \chi (t,x)\nabla_x( V(x)+V^\e_B(t,x))\cdot \nabla_p \sigma(p)  \beta^\e(t, dx, dp) dt \\
& =  - \int_{\bbR^{d}}\chi (0,x) \sigma(u_0^\e(x)) \rho_0^\e(x)\, dx.
\end{align*}
First, consider the term involving $\nabla_x V^\e_B$: Having in mind the result of Lemma \ref{ppf1}, we are allowed to consider
$\chi(t,x) \nabla_x V^\e_B(t,x) \cdot \nabla\sigma(u^\e(t,x))$ as a test-function which is integrable with repsect to $\rho^\e(t,x)$. Thus, we can apply the push-forward formula \eqref{pf1} and infer
\begin{align*}
&\int_0^\infty\iint_{\bbR^{2d}}\chi(t,x) \nabla_x V^\e_B(t,x)\cdot \nabla \sigma(p) \beta^\e(t,dx,dp) \, dt= \\
& = \int_0^\infty\int_{\bbR^d}\chi(t, X^\e(t,x))
\nabla_x V^\e_B(t,X^\e(t,x)) \cdot \nabla\sigma(u^\e(t,X^\e(t,x))) \rho_0^\e(x) \, dx \, dt.
\end{align*}
In addition, the fact that $J^\e=\rho^\e u^\e \in L^1(\R^d)$ implies that the ``test-function" $\sigma(u^\e(t,x)) u^\e(t,x) \cdot \nabla \chi(t,x) $ is integrable with respect to $\rho^\e$ (note however, that 
$u^\e(t,x)$ is not necessarily continuous). Thus we can again apply the (generalized) push-forward 
formula  \eqref{pf1} to obtain
\begin{align*}
&\int_0^\infty\iint_{\bbR^{2d}} \sigma(p) p\cdot \nabla_x \chi (t,x) \beta^\e(t,dx,dp)\, dt = \\
& = \int_0^\infty\int_{\bbR^d}\sigma(u^\e(t,x)) u^\e(t,x) \cdot \nabla \chi(t,x)  \rho_0^\e(x) \, dx\, dt.
\end{align*}
All the other terms appearing in the weak formulation of \eqref{maineq} can then be treated analogously. Having in mind the ODE system \eqref{bohm2}, we consequently arrive at 
\begin{align*}
& \int_0^\infty \iint_{\bbR^{2d}} \left( (\partial_t \chi  + p\cdot \nabla_x \chi)\sigma(p) - \chi \nabla_x( V+V^\e_B)\cdot \nabla_p \sigma \right) \beta^\e(t, dx, dp) dt \\ 
& \ =  \int_0^\infty \int_{\bbR^{d}} \frac{d}{dt} \left(\chi (t, X^\e(t,x)) \sigma(u^\e(t,x)) \right)\rho_0^\e(dx)\, dt\\
& \ = - \int_{\bbR^{d}}\chi (0,x) \sigma(u_0^\e(x)) \rho_0^\e(dx).
\end{align*}
This proves that the Bohmian measure $\beta^\e(t)$ furnishes a weak solution of \eqref{maineq} in $\mathcal D'([0, \infty) \times \R^d_x\times \R^d_p)$ with initial data \eqref{ini}. The proof for $\mathcal D'(\R_t \times \R^d_x\times \R^d_p)$ is analogous.
\end{proof}

\section{Study of possible defects}\label{sec: def}

Having established the fact that $\beta^\e$ is indeed a weak solution of \eqref{maineq}, we first rewrite the equation in the following form
\be\label{rewrit}
 \partial_t \beta^\e + p \cdot \nabla_x \beta^\e - \nabla_x V \cdot \nabla_p \beta^\e = \diver_p (\nabla_x V_B^\e \beta^\e) ,
\ee
where $V_B$ is the Bohm potential defined in \eqref{bohmpot}. Using the weak convergence results given in \cite{MPS} we can pass to the limit on the left hand side of this equation (up to extraction of sub-sequences) 
in order to obtain 
\begin{equation}\label{limiteq}
 \partial_t \beta + p \cdot \nabla_x \beta - \nabla_x  V  \cdot \nabla_p \beta =  \mathcal F
\end{equation}
where the defect $\mathcal F$ is, as defined in \eqref{F}.
In order to gain some information on $\mathcal F$ (and to prove Theorem \ref{thneu}), we first derive from \eqref{limiteq} the following equation for the first moment of $\beta(t)$ with respect to $p\in \R^d$:
\be\label{limmom}
 \partial_t J + \diver_x \int_{\R^d_p} p \otimes p  \, \beta(t,x, dp) - \rho \nabla_x  V   = -\int_{\R_p^d} p \mathcal F(t,x,p) dp,
\ee
where $\rho, J$ denote the classical limits of $\rho^\e, J^\e$ given by \eqref{limdensities} and \eqref{limdensities1}. Note that the integral on the ride hand side is well-defined 
since energy conservation implies 
\[
\int_{\R^d} |p|^2 \beta^\e(t,x, dp) \leq C,
\]
where the constant $C>0$ is independent of $\e$. This allows us to conclude that all terms on the left hand side of \eqref{limmom} are well-defined in the sense of distributions on $\R_t \times \R_x^d$ 
and hence also the right hand side is. To proceed further we need the following result.

\begin{lemma} \label{lemCS} 
Let $\mu \in \mathcal M^+(\R^d_x\times \R^d_p)$ be such that 
\[
\iint_{\R^{2d}} |p|^2 \mu(dx, dp)  < +\infty, 
\]
and
\[
 \int_{\R^d_p} \mu (x, dp) = \rho\in \mathcal M^+(\R^d), \quad \int_{\R^d_p} p \mu (x, dp) = \rho(x) u(x),
\]
for some function $u(x)\in \R^d$ defined $\rho$ - $a.e.$. Then it holds
\[
\int_{\R^d_p} p \otimes p \,  \mu(x, dp) \geq \rho(x) \, u(x) \otimes u(x),
\]
with equality if and only if $\mu(x,dp) = \rho(x)\delta (p-u(x))$.
\end{lemma}
\begin{proof} 
The proof follows directly from the fact that for $\rho$ - $a.e.$ $x$ we have
\[
0 \leq \int_{\R^d_p} (p-u(x)) \otimes (p - u(x)) \mu(x, dp) = \int_{\R^d_p} p \otimes p \mu (x, dp) - \rho u \otimes u,
\]
with equality if and only if $\mu(x,dp) = \rho(x) \delta (p-u(x))$.
%The proof follows directly from the Cauchy-Schwarz inequality (see also Lemma 3.5 in \cite{GaMa}) applied to 
%\[
%\sum_{\ell, j=1}^d \int_{\R^d} \rho(x) \varphi_\ell(x) \varphi_j(x) u_\ell(x) u_j(x) dx,
%\]
%where $\varphi\in C_{\rm c}^\infty(\R^d)$. Equality then holds if and only if there exists a constant $C\in \R$, such that 
%\[
%\forall \varphi \in C_{\rm c}^\infty: \ \sum_{\ell, j=1}^d p_\ell \varphi_j(x) = C \sum_{\ell, j=1}^d u_\ell(x) \varphi_j(x) , \quad \text{$\mu$ - $ a.e.$}
%\]
%which clearly implies that $p_\ell = C u_\ell(x)$ and thus for any $x\in \R^d$, $\mu(x,\cdot)$ must concentrate on the set of points $\{Cu_\ell(x)\}_{\ell =1}^d$.
\end{proof}
Using the result of this lemma we can rewrite \eqref{limmom} as 
\be
 \partial_t J + \diver_x (\rho  u \otimes u) + \rho \nabla_x  V   = -\int_{\R_p^d} p \mathcal F(t,x,p) dp - \diver_x(\rho \mathcal B),
\ee
with a defect $\mathcal B(t,x) \ge 0$. In addition we know that 
\[
\mathcal B(t,x) = 0 \mbox{, if and only if, $\beta(t,x,p) = \rho (t,x)\delta(p-u(t,x))$.}
\]
On the other hand, we can consider the equation for $\beta^\e(t)$, take first the moment w.r.t. $p$ and then pass to the limit $\e \to 0_+$: Multiplying \eqref{rewrit} by $p\in \R^d$ and integrating yields 
the equation for the current density in the quantum hydrodynamical system \eqref{qhd}, i.e. 
\be
\partial_t J^\e + \diver \left(\frac{J^\e \otimes J^\e}{\rho^\e} \right) + \rho^\e \nabla V = \rho^\e  \nabla V^\e_B,
\ee
where we have used the fact that
\[
\int_{\R^d} p \otimes p \beta^\e (t,x,dp) = \frac{J^\e \otimes J^\e}{\rho^\e}, 
\]
since  $\beta^\e$ is mono-kinetic by definition. Using this, we can define a defect $\mathcal C(t,x) \ge 0$ via 
\be\label{C}
\lim_{\e \to 0_+}\left(\frac{J^\e \otimes J^\e}{\rho^\e} \right) =   \int_{\R^d} p \otimes p \beta (t,x,dp) + \mathcal C(t,x),
\ee
where the limit has to be understood in $\mathcal D' (\R_t\times \R^d_x)$.
In addition, we know that 
\[
\rho^\e\nabla V_B^\e=\frac{\e^2}2\nabla \Delta \rho^\e- {\e^2}\diver\left( \nabla \sqrt{\rho^\e} \otimes \nabla \sqrt{\rho^\e} \right),
\]
where the first term tends to zero as $\e \to 0_+$ by linearity. This consequently yields
\be \label{limbohm}
\partial_t J + \diver (\rho u\otimes u) + \rho \nabla V = - \diver(\mathcal A + \rho \mathcal B + \mathcal C),
\ee
where $\mathcal A(t,x)$ is defined by
\be \label{A}
 {\e^2}\left( \nabla \sqrt{\rho^\e} \otimes \nabla \sqrt{\rho^\e} \right)  \stackrel{\e\rightarrow 0_+
}{\longrightarrow}\mathcal A \quad \text{in $\mathcal D'(\R_t\times \R^d_x)$.}
\ee
In summary, we have the following partial characterization of $\mathcal F$.

\begin{lemma} \label{lemchar} The defect $\mathcal F$ defined in \eqref{F} satisfies
\be \label{Fchar}
\int_{\R_p^d} p \mathcal F(t,x,p) dp = - \diver( \mathcal A(t,x) + \mathcal C(t,x)),
\ee
where $\mathcal A$, $\mathcal C$ are given by \eqref{A}, \eqref{C}, respectively.
\end{lemma}

In a last step, this can now be compared with the classical limit of the quantum hydrodynamic system \eqref{qhd} via Wigner measures. In  \cite{GaMa} it has been shown that 
$$
J^\e(t,x) \stackrel{\e\rightarrow 0_+
}{\longrightarrow}J(t,x) := \int_{\R^d} w(t,x,dp)
$$
which satisfies
\be \label{limwig}
\partial_t J + \diver (\rho u\otimes u) + \rho \nabla V = - \diver \rho \mathcal T,
\ee
with a temperature tensor $\mathcal T(t,x) \ge 0$. The latter is found to be equal to zero, if and only if $w(t,x,p) = \rho(t,x) \delta(p-u(t,x))$. This can now be used as follows:

\begin{proof}[Proof of Theorem \ref{thneu}] 
Let $d=1$. Having in mind that, by assumption, the initial limiting Bohmian and Wigner measures are equal $\beta_0(x,p) = w_0(x,p)$, the 
uniqueness of solutions, together with \eqref{limwig} and \eqref{limbohm}, implies
\be\label{Teq}
\rho \mathcal T = \mathcal A + \rho \mathcal B + \mathcal C,\quad \mbox{in $\mathcal D'(\R_t \times \R^d_x)$},
\ee
Since all terms on the right hand side are greater or equal to zero, we infer that 
\[
w(t,x,p) = \rho(t,x) \delta(p-u(t,x)), \mbox{ if and only if, 
$\mathcal A =  \mathcal B = \mathcal C = 0$.}
\] 
By definition, this implies that 
$$
\rho^\e\nabla V_B^\e \stackrel{\e\rightarrow 0_+
}{\longrightarrow} 0,
$$
as well as 
$$
\lim_{\e \to 0_+} \int_{\R^d} p \otimes p \beta^\e (t,x,dp) =   \int_{\R^d} p \otimes p \beta (t,x,dp)  = \rho u \otimes u.
$$
By Lemma \ref{lemCS}, we conclude $\beta(t,x,p) = \rho (t,x)\delta(p-u(t,x))$ and the assertion is proved.
\end{proof}
\begin{remark} In dimensions $d>1$ we can not conclude as before, since identity \eqref{Teq} has to be replaced by 
\[
\rho \mathcal T = \mathcal A + \rho \mathcal B + \mathcal C + \mathcal D,
\]
for some $\mathcal D$ satisfying $\diver \mathcal D(t,x) = 0$.
\end{remark}

\section{Bohmian measures for semi-classical wave packets}\label{sec: coh}

The theory of semi-classical wave packets is very well developed, see e.g. \cite{Ha, HaJo, Pa1, Pa2} and the references given therein 
(see also \cite{APPP, CaFe} for a recent application in the context of nonlinear Schr\"odinger equations). It allows us to approximate the solution to \eqref{sch} via
\be\label{coherent}
\psi^\e(t,x) \stackrel{\e\rightarrow 0_+
}{\sim} u^\e(t,x)=\e^{-d/4} v\left(t,\frac{x -X(t)}{\sqrt{\e}} \right) e^{i (P(t) \cdot (x- X(t)) +S(t))/\e},
\ee
where $P(t), X(t)$ solve the Hamiltonian system \eqref{classical} and $S(t)$ is the associated classical action, i.e.
\[
S(t) = \int_0^t \frac{1}{2} |P(s)|^2- V(X(s)) ds.
\]
The envelope function $v(t,y)$ is thereby found to be a solution of the following $\e$-independent Schr\"odinger equation (see also the proof of Theorem \ref{th2} below):
\be\label{veq}
i \partial_t  v = -\frac{1}{2}\Delta_y v+ \frac{1}{2} \,  (Q(t)y,y) v ,\quad
v(t=0,y)=a(y),
\ee
where $Q(t):= \text{Hess}\, V(X(t))$
denotes the Hessian of the potential $V(x)$ evaluated at the classical trajectory $X(t)$ and $a\in \mathcal S(\R^d)$ is induced by the initial data $\psi^\e_0$ given in Theorem \ref{th2}. In other words, $v(t,x)$ solves a 
linear Schr\"odinger equation with time-dependent quadratic potential. Under suitable assumptions on $V$ (satisfied by the hypothesis of Theorem \ref{th2}), one can show, see e.g. 
\cite{APPP,Ha, HaJo, CaFe, Pa1, Pa2}, that $u^\e(t,x)$ approximates the exact solution $\psi^\e(t,x)$ of \eqref{sch} in the following sense
\be\label{cohest}
\| \psi^\e (t,\cdot) - u^\e(t,\cdot) \|_{L^2(\R^d)} \leq C_1 \sqrt{\e} e^{C_2 t},
\ee
provided the initial data $\psi^\e(0,x)$ is of the form given in Theorem \ref{th2}. In other words, the approximation is valid up to times of order $ \ln (1/\e)$, at most.
%can be rephrased as follows: There exists a $c>0$ such that:
%\[
%\sup_{0\leq t \leq c \ln (1/\e)} \| \psi^\e (t, \cdot) - u(t, \cdot) \|_{L^2(\R^d)} \stackrel{\e \to 0}{\longrightarrow} 0,
%\]
%as long as: $c\, C_2  < 1/2$.
\begin{remark} Note that in contrast to the WKB approximation, the coherent state ansatz does not suffer from 
the appearance of caustics (although it is sensitive to them through equation \eqref{veq} where the caustics are somehow
hidden). In addition, it assumes 
that the amplitude concentrates on the scale $\sqrt{\e}$. The latter has been shown to be a critical scaling in the theory of Bohmian measures, cf. the case studies in \cite{MPS}.
\end{remark}

\begin{proof}[Proof of Theorem \ref{th2}] We first note that the solution to \eqref{veq} satisfies $\|v(t,\cdot ) \|_{L^2} = \| a \|_{L^2}$ for all $t\in \R$. Thus, the Wigner transformation of  $u^\e(t,x)$ 
satisfies 
$$
w^\e[u^\e]  \stackrel{\e\rightarrow 0_+
}{\longrightarrow} w \quad \text{in $L^\infty(\R_t; \mathcal S'(\R^d_x \times \R^d_p)) \, 
%{\rm w}-\ast$},
{\rm weak}^\ast$}.
$$
The corresponding Wigner measure is well known, cf. \cite{LiPa, MPS}:
$$
 w(t,x,p) =\| a\|_{L^2}^2 \, \delta(x-X(t)) \delta (p -P(t)).
$$
From the estimate \eqref{cohest} and the classical results given in \cite{LiPa}
we conclude that the Wigner transformation of the exact solution $w^\e[\psi^\e]$ 
converges to the same limiting measure $w$, uniformly on compact time-intervals $I \subset \R_t$. 

In order to prove that $w(t) = \beta(t)$, 
we perform the unitary transformation: $\psi^\e\mapsto v^\e$, defined via 
\be\label{trafo}
\psi^\e(t,x) = \e^{-d/4} v^\e\left(t,\frac{x -X(t)}{\sqrt{\e}} \right) e^{i (P(t) \cdot (x- X(t)) +S(t))/\e}.
\ee

Using this transformation, equation \eqref{sch} is easily found to be equivalent to 
\be\label{vepseq}
i \partial_t  v^\e = -\frac{1}{2}\Delta_y v^\e+ V^\e(t,y)v^\e ,\quad v^\e(t=0,x)   = a(x),
\ee
where $V^\e(t,y)$ is given by 
\[
V^\e(t,y) = \frac{1}{\e} \left (V(X(t)+\sqrt{\e}y) - V(X(t)) - \sqrt{\e} \nabla V(X(t))\cdot y \right).
\]
Obviously, for $C^2$ potentials $V$ equation \eqref{vepseq} converges to \eqref{veq} as $\e \to 0_+$. 
This together with sufficient a-priori bounds on $v^\e(t)$ yields the estimate \eqref{cohest}, cf. \cite{CaFe} for more details. On the other hand, using \eqref{trafo}, the 
Bohmian measure $\beta^\e(t)$ of the exact solution $\psi^\e(t)$ can be seen to act on Lipschitz test-function $\varphi\in C_0(\R^d_x\times \R^d_p)$ via
\begin{align*}
 \langle \beta^\e(t), \varphi\rangle  = \int_{\R^d} | v^\e(t,y)|^2 \varphi\left(X(t) + \sqrt{\e} y, \sqrt{\e} \im \left(\frac{ \nabla v^\e(t, y)}{ v^\e(t, y)} \right ) +P(t)\right) \, dy.
\end{align*}
Using the Lipschitz continuity of $\varphi$ we can estimate
\begin{align*}
&  \Big |\varphi\left(X(t) + \sqrt{\e} y, \sqrt{\e} \im \left(\frac{ \nabla v^\e(t, y)}{ v^\e(t, y)} \right ) +P(t)\right)  - \varphi(X(t), P(t))\Big |\\
& \ \leq C_\varphi\sqrt{\e} \left(\vert y\vert 
  + \Big |\im \left(\frac{ \nabla v^\e(t, y)}{ v^\e(t, y)} \right)\Big |\right),  
\end{align*}
for some positive constant $C_\varphi >0$. In view of this, we obtain
\begin{align*}
 & \ \Big | \langle \beta^\e(t), \varphi\rangle - \int_{\R^d} | v^\e(t,y) | ^2  \varphi  (X(t), P(t)) \, d y \Big |   \\
 & \ \leq  C_\varphi\sqrt{\e}  \int_{\R^d} |y|  | v^\e(t,y)|^2 dy +  \sqrt{\e}\int_{\R^d}   | v^\e(t,y)| | \nabla v^\e(t,y)| dy \\
  & \ \leq  C_\varphi\sqrt{\e}  \,  \|v^\e(t) \|_{L^2} \big( \| \, |y| v^\e(t) \| _{L^2} + \ \|\nabla v^\e(t) \|_{L^2}  \big) ,
\end{align*}
where the last inequality directly follows from Cauchy-Bunjakowski-Schwarz. In order to proceed further we need the following lemma.
\begin{lemma}\label{a-priori} Let $V\in C^3_{\rm b}(\R^d)$. Then the solution of \eqref{vepseq} satisfies 
\[
\| \, |y| v^\e(t) \| _{L^2} \leq C_1 ,\quad  \|\nabla v^\e(t) \|_{L^2}  \leq C_2, \quad \forall \, t\in  [-T, T]\subset \R ,
\]
where $C_1, C_2$ are some positive constants, independent of $\e$.
\end{lemma}
\begin{proof}[Proof of Lemma \ref{a-priori}]
In \cite{CaFe} it is shown  that, if 
$V$ is quadratically bounded, i.e. $\partial^\gamma V (x) \in L^\infty$, for all $|\gamma| \ge 2$, it holds:
\be \label{estCa}
\| \, |y| v^\e(t) \| _{L^2} \le C_1 , \quad \|\, |y|^3 v^\e(t) \|_{L^2} \leq C_3, \forall \, t\in \R.
\ee
It therefore only remains to show the estimate for $\nabla v^\e(t)$. This follows by considering the energy corresponding to \eqref{vepseq}, i.e. 
\[
E^\e(t) =  \frac{\e^2}{2} \int _{\R^d} | \nabla v^\e (t,x) |^2 dx + \int_{\R^d} V^\e(t,x) |  v^\e (t,x) |^2 dx ,
\]
which satisfies
\[
\frac{d}{dt}  E^\e(t) = \int_{\R^d} \partial_t V^\e(t,x) |  v^\e (t,x) |^2 dx .
\]
Since 
$$|\partial_t V^\e(t,x) | \leq |\dot X(t)| |y|^3 \sqrt{\e} \sup |\partial^3 V (X(t) + s \sqrt{\e} y)|,
$$
the assumption $V\in C^3_{\rm b}$, together with \eqref{estCa}, yields the desired bound on $\nabla v^\e(t)$.
\end{proof}
Using the a-priori estimates established in Lemma \ref{a-priori} we obtain
\[
\Big | \langle \beta^\e(t), \varphi\rangle - \int_{\R^d} | v^\e(y,t) | ^2  \varphi  (X(t), P(t)) \, d y \Big |\stackrel{\e\rightarrow 0_+
}{\longrightarrow}0.
\]
In other words, we have that
\be\label{betalim}
\beta(t) = \|v(t,\cdot)\|_{L^2}^2 \delta(x-X(t))\delta(p-P(t)),
\ee
Having in mind that $\|v(t,\cdot ) \|_{L^2} = \| a \|_{L^2}$ this proves assertion (1) of Theorem \ref{th2}.\\

In order to conclude assertion (2) of Theorem \ref{th2} we recall the following formula, stated in \cite[Remark 3.8]{MPS}. For all $t\in \R$ and for all test-functions 
$\varphi \in C_0(\R^d_x\times\R^d_p)$, $\chi\in C_0(\R_t)$ it holds
\begin{equation}\label{formula}
\begin{split}
&  \int_{\R} \chi  (t) \iint_{\R^{2d}} \varphi(x,p) \beta^\e (t, dx, dp) dt= \\
& =  \int_{\R} \chi  (t) \int_{\R^{d}} \varphi(X^\e(t,x),P^\e(t,x))\rho_0^\e(x) \, dx \, dt \\
& =  \,  \int_{\R} \chi  (t) \int_{\R^{d}} \varphi(X^\e(t,x_0+\sqrt{\e} y),P^\e(t,x_0+\sqrt{\e} y))|a(y)|^2 dy \, dt ,
\end{split}
\end{equation}
where in the second equality we set $y= (x-x_0 )/ \sqrt{\e}$ and recall that the initial density is given by 
$$\rho^\e_0(x) = \e^{-d/2} \left| a\left(\frac{x-x_0}{\sqrt{\e}}\right)\right|^2.$$
Now, let the function $\zeta^\e(t,y)$ be defined as
$$
\zeta^\e: \R_t \times \R^d_y \ni (t,y) \mapsto (Y^\e(t,y), Z^\e(t,y))\in \R^d_x\times \R^d_p,
$$
where
$$
Y^\e (t,y) = X^\e (t, x_0 + \sqrt{\e} y), \quad Z^\e (t,y) = P^\e (t, x_0 + \sqrt{\e} y),
$$
are the rescaled Bohmian trajectories. The associated Young measure (or, parametrized measure):
$$\omega_{t,y}:\R_t\times\R^d_y \to \mathcal M^+( \R^d_x\times \R^d_p)\, ; \ (t,y)\mapsto \omega_{t,y}(x,p),$$
is defined, after extraction of a subsequence, via
$$
\lim_{\e \to 0} \iint_{\R_t\times \R_y^d} \sigma ( t, y, \zeta_\e(t,y)) \, dy \, dt = \iint_{\R_t \times \R_y^d} \iint_{\R^{2d}} \sigma ( t,y,x,p) d\omega_{t,y}(x,p) \, dy\,  dt,
$$
for any test function $\sigma \in L^1(\R_t\times \R^d_y; C_0(\R^{2d}))$, cf. \cite{Ba, Pe1, Pe2}. 

By passing to the limit $\e \to 0_+$ in \eqref{formula} we find that (after the choice of an appropriate sub-sequence):
\begin{align*}
 \int_{\R} \chi  (t) \iint_{\R^{2d}} \varphi(x,p) \beta (t, dx, dp) dt =   \int_{\R} \chi  (t) \int_{\R^{d}} \varphi(x,p)\omega_{t,y} (dx, dp) |a(y)|^2 dy \, dt ,
\end{align*}
and thus
\[
\beta(t,x,p)= \int_{\R^d} |a(y)|^2 \omega_{t,y} (x,p) dy.
\]
Upon inserting \eqref{betalim} with $\|v(t,\cdot ) \|_{L^2} = \| a \|_{L^2}$ , this implies 
$$
\omega_{t,y}(x,p) = \nu (t,y) \delta(p-P(t)) \delta(x-X(t)).
$$
Since $0\leq \nu(y,t) \le 1$, $|a(y)|^2>0$, by assumption, and
\[
\int_{\R^d} |a(y)|^2 dy = \int_{\R^d} |a(y)|^2 \nu(t,y) dy,
\]
for all $t\in \R$, we conclude $\nu(t,y) \equiv 1$ a.e. with respect to $|a(y)|^2dy$ and hence
$$
\omega_{t,y}(x,p) =  \delta(p-P(t)) \delta(x-X(t)).
$$
In other words, $\omega_{t,y}$ is found to be concentrated in a single point in phase space $\R^d_x\times \R^d_p$.
By a well known result of Young measure theory (see e.g. \cite[Proposition 1]{GS}), we know that the fact that $\omega_{t,x}$ is concentrated in a point is equivalent to the convergence of the re-scaled trajectories, i.e.
$$
Y^\e  \stackrel{\e\rightarrow 0_+ }{\longrightarrow} X, \quad Z^\e  \stackrel{\e\rightarrow 0_+ }{\longrightarrow} P,
$$
locally in measure on $\R_t \times \R^d_y$.\end{proof}

\begin{remark}\label{rem: last}
In $d=1$, assertion (1) of Theorem \ref{th2} directly follows from Theorem \ref{thneu}, since $w(t)$ is obviously mono-kinetic. In addition, one should note that the established 
local in measure convergence implies (see e.g. \cite[Section 13]{Do}) that there exists a sub-sequence $\{ \e_n \}_{n \in \mathbb N}$, going to zero as $n\to \infty$, such that 
\[Y^{\e_n}  \stackrel{n\rightarrow \infty }{\longrightarrow} X, \quad Z^{\e_n}  \stackrel{n \rightarrow \infty }{\longrightarrow} P, \quad \mbox{a.e. in $\Omega \subseteq \R_t\times \R^d_y$.}
\]
\end{remark}

The rescaling of the position variable $x= x_0 + \sqrt \e y$ in the Bohmian trajectories is critical for the coherent wave packet of Theorem \ref{thneu} in the sense that it
allows, at any time $t\in \R$, to directly connect the Bohmian measure associated to the wave function with the Young measure of the Bohmian trajectories (with rescaled initial position).
It seems that this cannot be done for any other rescaling of $u^\e$ (e.g. replacing $\sqrt \e$ by $\e^\alpha$).\\
%Also note that for any other rescaling the corresponding wave packet $u^\e(t,x)$ in general will not retain its form during the time-evolution.

{\bf Acknowledgment.} The authors want to thank the referees for helpful remarks in improving the paper and in particular for the short, direct argument given in the proof of Lemma \ref{lemCS}.

%%%%%%%%%%%%%%%%%%%%%%%%%%%%%%%%%%%%%%%%%%%%%%%%%

\appendix

\section{An example with non mono-kinetic limiting Bohmian measure}

In \cite{MPS} we considered several different examples of $\psi^\e$ and computed the corresponding limiting Bohmian measure $\beta$ and the corresponding Wigner measure $w$. We found that in general 
$w\not = \beta$ except in rather special situations. In fact, in all the examples given in \cite{MPS} 
we find $w= \beta$ only in the mono-kinetic case. 
Together with the results stated in Theorem \ref{thneu} and Theorem \ref{th2} this might yield the wrong impression that $w$ and $\beta$ can only coincide if they are both mono-kinetic phase space distributions. 
The following example will illustrate that this is 
in general not the case: 

Consider an $\e$-dependent family of wave functions $\{u^\e\}_{0<\e\leq1}$ given by
\[
u^\e(x)=a^\e(x)e^{i S(x)/ \e},
\]
where the amplitude $a^\e$ reads
\[
a^\e(x) = \e^{-d/4}\rho^{1/2} \left(\frac{\vert x\vert}{\e^{1/2}}\right), 
\]
with some $\e$-independent profile $\rho\in\mathcal S(\bbR;\R)$, i.e. smooth and rapidly decaying.
%, satisfying 
%\be\label{assrho}
% \int_0^\infty \rho(r)\, r^{d-1}dr<+\infty.
%\ee
%\be\label{assrho}
% \int_0^\infty \frac{(\rho'(r))^2}{{\rho(r)}}\, r^{d-1}dr<+\infty.
%\ee
In addition, let $S\in C_{\rm b}(\bbR^d)\cap C^2_{\rm b}(\bbR^d/\{0\})$, such that $S(x)$ is even and has a cone-like singularity at $x=0$:
$$
\lim_{\delta\to 0}\nabla S(\delta\omega)=\chi(\omega),\ \forall\omega\in \mathbb S^{d-1}\, ,
$$
for $\chi\in C^\infty(\mathbb S^{d-1})$. 
%We then define: $S^\e(x) = S(x)$, for $\vert x\vert>\e^{3/4}$.
%On the other hand, for $\vert x\vert\leq\e^{3/4}$ we define the phase function $S^\e(x)$ through
%$$
%\nabla S^\e(\e^{3/4}t\omega)=\nabla S(\e^{3/4}\omega)\frac{1+t}2+\nabla S(-\e^{3/4}\omega)\frac{1-t}2,\quad -1 \leq t \leq 1.$$

\begin{lemma} Let $u^\e$ be as given above, then
\[
\beta(x,p)=w(x,p) =\frac 1{\vert \mathbb S^{d-1}\vert}\int_\R \rho (\vert y\vert)dy\, \int_{\mathbb S^{d-1}}\delta(p-\chi(\omega))d\omega\otimes \delta(x).
\]
\end{lemma}
\begin{remark} To our knowledge this is the first example in which the projection of $\beta = w$ onto position space 
is absolutely continuous with respect to the Lebesgue measure on $\R^d_p$.\end{remark}
\begin{proof} We first note that, by assumption,
\[
\vert u^\e(x)\vert^2=\e^{-d/2}\rho \left(\frac{\vert x\vert}{\e^{1/2}}\right)\stackrel{\e\rightarrow 0_+
}{\longrightarrow}\delta(x)\int_{\bbR}\rho(\vert y
\vert)dy.
\]
Moreover, it has been already computed in \cite{GaMa}, that 
\[
\lim_{\e \to 0_+}\langle w^\e[u^\e],\vp \rangle  =\frac 1{\vert \mathbb S^{d-1}\vert} \int_\R \rho (\vert y\vert)dy \int_{\mathbb S^{d-1}} \vp(0, \chi(\omega)) d\omega,
\]
for any test function $\varphi\in C_0(\R^{2d})$. Thus, we see that the Wigner measure $w(x,p)$ is of the form given above.
%In addition, we also have that $\nabla S^\e\in W^{1,\infty}(\bbR^d)$ and $ \vert\partial_{\ell}\partial_{j}S^\e\vert\leq C\e^{-3/4}$,  for all $ \ell, j =1, \dots, d$. Thus 
%\[
%\e \Vert \partial_{\ell}\partial_{j}S^\e\Vert_{L^\infty(\bbR^d)}\stackrel{\e\rightarrow 0_+
%}{\longrightarrow} 0, \quad \forall \, \ell, j =1, \dots, d.
%\]
%Next, we consider $\e^2\vert\nabla u^\e \vert^2=\e^2\vert \nabla a^\e \vert^2+\rho^\e \vert\nabla S^\e\vert^2$, where we denote, as usual $\rho^\e := | u^\e|^2$. Since $\vert \nabla S^\e(x)\vert\leq C$, by assumption, we infer 
%\[
%\int_{\bbR^d}\rho^\e(x)\vert\nabla S^\e(x)\vert^2dx\leq C, \mbox{ uniformly in }\e.
%\]
%On the other hand one easily computes
%\[
%\e^2\int_{\R^d} \vert \nabla a^\e (x) \vert^2=
%\frac\e 4\int_\R \frac{\rho'(\vert y\vert)^2}{{\rho(\vert y\vert)}}dy\stackrel{\e\rightarrow 0_+
%}{\longrightarrow} 0,
%\]
%in view of \eqref{assrho}. Theorem 4.7 of \cite{MPS} consequently implies $\beta(x,p)=w(x,p)$ in the sense of measures. 

It remains to explicitly compute the limiting Bohmian measure. To this end, we consider the action of $\beta^\e$ onto any 
test function $\varphi(x,p)$, i.e.
\begin{align*} 
\langle\beta^\e,\vp\rangle = &\  {\e^{-d/2}}\int_{\R^d} \rho\left( \frac{\vert x\vert}{\e^{1/2}}\right )\vp(x,\nabla S(x))dx\\
= & \ \int_{\mathbb S^{d-1}}\int_0^\infty\rho(r)\vp(\e^{1/2}r\omega,\nabla S(\e^{1/2}r\omega))r^{d-1}dr d\omega,
\end{align*}
by setting $y=r\omega$. It easily follows that as $\e \to 0_+$:
\[
\langle\beta^\e,\vp\rangle  \sim
\int_{\mathbb S^{d-1}}\int_0^\infty\rho(r)r^{d-1}\vp(0,\nabla S(\e^{1/2}r\omega)drd\omega.
\]
Keeping $r>0$, $\omega\in \mathbb S^{d-1}$ fixed, we see that for $\e$  sufficiently small,
\[
\nabla S(\e^{1/2}r\omega)\stackrel{\e\rightarrow 0_+
}{\longrightarrow}\chi(\omega).
\]
%since $\e^{1/2}\gg \e^{3/4}$. 
By dominated convergence, we therefore conclude 
\[
\langle\beta^\e,\vp\rangle \stackrel{\e\rightarrow 0_+
}{\longrightarrow} \frac 1{\vert \mathbb S^{d-1}\vert}\int_\R \rho(\vert y\vert)dy\int_{\mathbb S^{d-1}}\vp(0,\chi(\omega))d\omega,
\]
and the assertion is proved.
\end{proof}

\end{document}